\documentclass{article}
\usepackage{amsmath, amsthm, amssymb, braket, enumitem, graphicx, subcaption, cite}
\usepackage[colorlinks=true]{hyperref}
\usepackage[margin=2cm]{geometry}

\DeclareMathOperator{\tr}{Tr}

\newtheorem{theorem}{Theorem}

\newtheorem{corollary}{Corollary}

\newtheorem{definition}{Definition}

\newtheorem{lemma}{Lemma}

\DeclareMathOperator{\Supp}{Supp}

\title{$r$-deformed $\alpha$-$z$-R\'enyi relative entropy}
\author{
	Srikrishna Maity $^a$\thanks{Email: \texttt{krishmaitybb@gmail.com}}, Shigeru Furuichi $^{b,c}$ \thanks{Email: \texttt{furuichi.shigeru@nihon-u.ac.jp}}, Supriyo Dutta $^a$ \thanks{Email: \texttt{dosupriyo@gmail.com}}
	\vspace{.3cm}\\
	$^a$\small{Department of Mathematics, National Institute of Technology Agartala} \\ \small{Jirania, West Tripura, Tripura, India - 799046.}
	\vspace{.3cm}\\
	$^b$\small{Department of Information Science College of Humanities and Sciences, Nihon University} \\
	\small{3-25-40, Sakurajyousui, Setagaya-ku, Tokyo, 156-8550, Japan.}
	\vspace{.3cm}\\
	$^c$\small{Department of Mathematics, Saveetha School of Engineering,}\\
	\small{SIMATS Thandalam, Chennai 602105, Tamilnadu, India.}
}
\date{} 

\begin{document}
	\maketitle 
	
	\begin{abstract}
		In this article, we consider the $r$-logarithm for defining three-parameter family of R\'{e}nyi relative entropies that are generalization of the $\alpha$-$z$-R\'{e}nyi relative entropies. All the members of $r$-deformed $\alpha$-$z$-R\'{e}nyi relative entropies satisfy the necessary axioms to be a divergence. We expose the range of parameters $\alpha$, $z$ and $r$ for which the data processing inequality holds. We also establish that $r$-deformed $\alpha$-$z$-R\'{e}nyi relative entropy is an upper bound of the Tsallis relative entropy. Now, we have two upper bounds of the Tsallis relative entropy, which are $r$-deformed $\alpha$-$z$-R\'{e}nyi relative entropy and the other one, which is discussed in literature \cite{fujii2021matrix}. We investigate the order relationship between these two upper bounds of the Tsallis relative entropy. We observe that our new upper bound is more tighter when applicable to the density operators.
		
	\end{abstract}
	
	
	\section{Introduction}
	
		Let $P = (p_1, p_2, \dots p_n)$ and $Q = (q_1, q_2, \dots q_n)$ be discrete probability distributions, for simplicity. The R\'{e}nyi divergence is defined by 
		\begin{equation}
			R_\alpha(P||Q) = \frac{1}{\alpha - 1} \log \left(\sum_{i = 1}^n p_i^\alpha q_i^{1 - \alpha} \right),
		\end{equation}
		where $\alpha > 1$ \cite{van2014renyi}. In quantum information theory, we substitute the classical R\'{e}nyi divergence by quantum R\'{e}nyi divergence. Now the classical definition of the R\'{e}nyi divergence is generalized for density operators. The quantum R\'{e}nyi $\alpha$-relative entropy \cite{shao2017quantum} between two density operators $\rho$ and $\sigma$ is defined by 
		\begin{equation}\label{Quantum_Rényi_alpha_relative_entropy} 
			R_{\alpha}(\rho || \sigma) = \frac{1}{\alpha - 1} \log \tr(\rho^{\alpha}
			\sigma^{1-\alpha}),
		\end{equation}
		for $0 \leq \alpha < 1$ and $\alpha > 1$. Later another parameter $z$ was added in this definition. Given any operator $\rho$, the subspace spanned by the set of eigenvectors corresponding to its non zero eigenvalues is denoted by $\Supp(\rho )$. Given a density operator $\rho$, and a positive semi-definite operator  $\sigma$ with $\Supp(\rho) \subseteq \Supp(\sigma)$, the $\alpha$-$z$ R\'{e}nyi quantum relative entropy \cite{audenaert2015alpha} is defined by
		\begin{equation}\label{z relative entropies}
			D_{\alpha,z}(\rho\|\sigma) = \frac{1}{\alpha-1} \log\tr\left( \sigma^{\frac{1-\alpha}{2z}} \rho^{\frac{\alpha}{z}} \sigma^{\frac{1-\alpha}{2z}} \right)^z,
		\end{equation}
		where $\alpha \in \mathbb{R} - \{1\}$ and $z \in \mathbb{R}^{+}$. Also, $D_{1, z}(\rho\|\sigma) = \lim_{\alpha \to 1} D_{\alpha,z}(\rho\|\sigma)$, $D_{\alpha, 0}(\rho\|\sigma) = \lim_{z \to 0} D_{\alpha,z}(\rho\|\sigma)$. There are positive semi-definite operators $\rho$ and $\sigma$ such that $\tr \left( \sigma^{\frac{1-\alpha}{2z}} \rho^{\frac{\alpha}{z}} \sigma^{\frac{1-\alpha}{2z}} \right)^z= 0$. The natural logarithm is undefined in these cases. Therefore, the $\alpha$-$z$ Rényi quantum relative entropy $D_{\alpha,z}(\rho\|\sigma)$ is also undefined. It motivates us to develop an alternative definition. In this article, we proposed the $r$-deformed $\alpha$-$z$ Rényi quantum relative entropy which overcome this drawback. Here, we deform the natural logarithm in $D_{\alpha,z}(\rho\|\sigma)$ with $r$-logarithm, which is defined as follows:\\
		The idea of deformed logarithms \cite{kac2002quantum}, entropy and divergence \cite{cover1999elements} is an interesting and well-investigated topic in literature \cite{kaniadakis2004deformed, leditzky2016relative, tsallis2019beyond, tsallis2009nonadditive, namdari2019review, dutta2021two}, due to its significance in Mathematics, Physics, Communication Engineering, and in Machine Learning.
		
		\begin{definition}\label{r_logarithm_definiton}
			The $r$-logarithm is denoted by $\ln_{r}(x)$ and it is defined by
			\begin{center}
				$\ln_r(x) = \frac{x^r - 1}{r}$ ,\quad for\; $r \neq 0$ and $x\geq 0$.
			\end{center}
		\end{definition}
		Note that, $\underset{r \to 0}{\lim} \ln_{r}(x) 
		= \log(x)$ for $x\neq 0$. Therefore, the $r$-logarithm is a generalized version of the natural logarithm. Also, it is defined for $x \geq 0$ which assists us to calculate the value of relative entropy for $\tr \left( \sigma^{\frac{1-\alpha}{2z}} \rho^{\frac{\alpha}{z}} \sigma^{\frac{1-\alpha}{2z}} \right)^z= 0$. Now, we have the following definition of the $r$-deformed $\alpha$-$z$-R\'{e}nyi relative entropy.
		\begin{definition}\label{r-z}
			The $r$-deformed $\alpha$-$z$-R\'{e}nyi relative entropy for any two positive semi-definite operators $\rho$ and $\sigma$ with $\Supp(\rho) \subseteq \Supp(\sigma) $ is defined by
			\begin{equation*}
				\begin{split}
					D_{\alpha,z}^{(r)}(\rho || \sigma)
					&= \frac{1}{\alpha - 1} \ln_{r} \left(\tr \left( \sigma^{\frac{1-\alpha}{2z}} \rho^{\frac{\alpha}{z}} \sigma^{\frac{1-\alpha}{2z}} \right)^z \right) 
					=\frac{1}{r(\alpha-1)} \left[ \left(\tr \left( \sigma^{\frac{1-\alpha}{2z}} \rho^{\frac{\alpha}{z}} \sigma^{\frac{1-\alpha}{2z}} \right)^z \right)^{r} - 1 \right],
				\end{split}
			\end{equation*}
			where $\alpha$, $r \in \mathbb{R} $ and $ z \in \mathbb{R}^{+}$.
		\end{definition}
		
		Now putting $r=z=1$ in \autoref{r-z}, we have 
		\begin{equation}
			D_{\alpha,1}^{1}(\rho || \sigma) = \frac{1}{(\alpha-1)} \left[\tr \left( \sigma^{\frac{1-\alpha}{2}} \rho^{\alpha} \sigma^{\frac{1-\alpha}{2}} \right) - 1 \right] =\frac{1-\tr\left(\rho^\alpha \sigma^{1-\alpha} \right) }{1-\alpha}.
		\end{equation} 
		Assuming $\rho$ as a density operator that is $\tr\left( \rho \right)=1$ we have
		\begin{equation}
			D_{\alpha,1}^{1}(\rho || \sigma) = \frac{\tr\left(\rho \right)-\tr\left(\rho^\alpha \sigma^{1-\alpha}\right)  }{1-\alpha}  =\tr\left( \frac{\rho-\rho^\alpha \sigma^{1-\alpha}}{1-\alpha}\right).
		\end{equation} 
		This is the Tsallis relative entropy, which is 
		\begin{equation}\label{equality_of_Tsalish_bound}
			D_{\alpha}\left( \rho||\sigma\right) = D_{\alpha,1}^{1}(\rho || \sigma) = \tr\left( \frac{\rho-\rho^\alpha \sigma^{1-\alpha}}{1-\alpha}\right).
		\end{equation}
		
		
		As the $r$-logarithm plays a key role in this article, we investigate the properties of $r$-logarithm. We prove that the $r$-deformed $\alpha$-$z$-R\'{e}nyi relative entropy satisfies the properties of divergence, which includes non-negativity, unitary invariance, pseudo-additivity, joint convexity, and the data-processing inequality.  
		Our fundamental observations discussed in this article are as follows:
		\begin{enumerate}
			\item 
				In \autoref{g}, we prove that $D_{\alpha,z}^{(r)}(\rho || \sigma)$ is an upper bound of Tsallis relative entropy $D_{\alpha}(\rho || \sigma)$ \cite{furuichi2004fundamental}, that is 
				 \begin{equation}
				  D_{\alpha}(\rho || 	\sigma)=\tr\left(\frac{\rho - \rho^{\alpha}\sigma^{1-\alpha}}{1-\alpha} \right) \leq D_{\alpha,z}^{(r)}(\rho\|\sigma)~~\text{for~}  r \in \mathbb{R} \text{~and~} \alpha \in \mathbb{R}-\{1\}.
				 \end{equation}
		 	\item 
		 		For $0 \leq \alpha < 1$ and $q > (1-\alpha)$, we compare two upper bounds $B_{1}(\rho || \sigma) = 	D_{\alpha,z}^{(r)}(\rho || \sigma)$ and 
			 \begin{equation}\label{bound definition}
			 	B_{2}(\rho || \sigma) =-\tr \left[\rho \ln_{1-\alpha}\left(\rho^{-q/2}\sigma^q\rho^{-q/2}\right)^{1/q}\right]
			 \end{equation}
			 of the Tsallis relative entropy \cite{fujii2021matrix, furuichi2010matrix, seo2019matrix}, numerically. We observe that for non-commutative density operators $\rho$ and $\sigma$,
			 \begin{equation}
			 D_{\alpha}(\rho || \sigma)	\leq B_{1}(\rho || 	\sigma)\leq B_{2}(\rho || \sigma).
			 \end{equation}
		\end{enumerate}
	
	Moreover we observe that the $r$-logarithm is advantageous for this generalization. We have already mentioned that the $r$-logarithm is defined at $x=0$ where the natural logarithm is not. We also observe that the use of $r$-logarithm makes the process of deriving a number of results simpler. For example the proof of the inequality (\ref{j4}) needs the operator mejorization inequalities. On the other hand the $r$-logarithm does not need this technique. This article is distributed as follows: In Section 2, we discussed about the properties of $r$-logarithm. The properties of $r$-deformed $\alpha$-$z$-R\'enyi relative entropy are explored in section 3. In section 4, we established a new upper bound of Tsallis relative entropy and shown a detailed comparison of this with another upper bound of Tsallis relative entropy. Then we conclude this article.

	\section{$r$-logarithm and its properties}
		
		
		Now, we have the following properties of the $r$-logarithm, which we utilize in this article for deriving the characteristics of $r$-deformed $\alpha$-$z$-R\'{e}nyi relative entropy.
		
		\begin{lemma}\label{Properties of r-logarithm 1}
			Let $x$ and $y$ be two real numbers, then $\ln_{r}(xy) = \ln_{r}(x) + \ln_{r}(y) + r\ln_{r}(x)\ln_{r}(y)$.
		\end{lemma}
		\begin{proof}
			Applying the definitions of $r$-logarithms, mentioned in \autoref{r_logarithm_definiton}, we get
			\begin{equation*}
				\ln_{r}(x)+\ln_{r}(y)+r\ln_{r}(x)\ln_{r}(x)
				=\frac{x^r-1}{r}+\frac{y^r-1}{r}+r\left( \frac{x^r-1}{r}\right) \left( \frac{y^r-1}{r}\right) 
				= \frac{(xy)^r-1}{r} = \ln_{r}(xy).
			\end{equation*}
			Hence, the proof.
		\end{proof}
		
		\begin{lemma}\label{Properties of r-logarithm 2}
			The $r$-logarithm satisfy the Jensen's inequality, which is
			\begin{equation}
				\ln_{r}\left( px+(1-p)y\right)\geq p\ln_{r}(x)+(1-p)\ln_{r}(y),~\text{for}~ 0\leq p \leq 1 ~\text{and}~ 0 < r \leq 1.
			\end{equation}
		\end{lemma}
		\begin{proof}
			Let $f(x)=\frac{x^r-1}{r}$. Then, $f''(x)=(r-1)x^{r-2} \leq 0$ when $0 < r \leq 1$ and $x > 0$. Therefore, $\frac{x^r-1}{r}$ is a concave function in this range of $r$.
		\end{proof}
			 
		\begin{lemma}\label{properties of r-logagithm 3}
			The $r$-logarithm is a monotone increasing function for $x \geq 0$.
		\end{lemma}
		\begin{proof}
			We have $\frac{d}{dx}\ln_{r}(x) = x^{r-1} \geq 0$, when $x > 0$. Therefore~$\ln_{r}(x)$ is a monotone increasing function for $x >0$. Also, $\ln_{r}(0)=-\frac{1}{r}< \frac{x^r-1}{r}=\ln_{r}(x)$ for all $x>0$ and $r>0$.
			
		\end{proof}

		\begin{lemma}\label{b}
			For $r<1$ and $ x \geq 0$, the $r$-logarithm $\ln_r(x) \leq x-1$.
		\end{lemma}
		\begin{proof}
			Let $g(x)=(x-1)-\ln_{r}(x)$. The proof is done if $g(x)\geq 0$ for $x \geq 0$. Note that, $g'(x)=1-x^{r-1}$. We have $r - 1<0$. Therefore, we have the following cases:
			\begin{enumerate}[label=\textbf{Case \arabic*:}, leftmargin=*]
				\item 
				For $0 < x \leq 1$ we get $x^{r-1} \geq 1$, or $1 - x^{r - 1} \leq 0$, or $g'(x) \leq 0$.
				\item 
				For $x > 1$ we get $x^{r-1} < 1$, or $1 - x^{r - 1} > 0$, or $g'(x)>0$.
			\end{enumerate}
			Hence, $g$ is increasing for $x>1$ and decreasing for $0<x \leq 1$. Now, $g''(x) = (1 - r)x^{r - 2}$. Also, $ g'(x)=0$ and $g''(x)>0$ at $x=1$. Thus, $g$ attains its minima, which is $0$ at $x=1$. Therefore, $g(x) \geq 0$ for $x > 0$. For  $x=0,~g(0)= -1- \ln_{r}(0)= -1 + \frac{1}{r} > 0$, because for this case $r$ must be in $(0,1).$
		\end{proof}	
%

		\begin{lemma}\label{a}
			For $r\geq 1$ and $ x \geq 0 $, the $r$-logarithm $\ln_r(x) \geq x-1 $.
		\end{lemma}
		\begin{proof}
			Let $h(x)=\ln_{r}(x)-(x-1)$. To prove this Lemma it is sufficient to prove $h(x) \geq 0 $ for $ x \geq 0$. Differentiating $h(x)$ we get $h'(x)=x^{r-1}-1$. We assumed that $r - 1 \geq 0$. Now, we have the following cases:
			\begin{enumerate}[label=\textbf{Case \arabic*:}, leftmargin=*]
				\item 
				For $0 < x \leq 1$, we get $ x^{r-1} \leq 1 $, or $ x^{r-1}-1 \leq 0$, or $h'(x) \leq 0 $.
				\item 
				For $ x>1 $, we get $ x^{r-1} > 1 $, or $ x^{r-1}-1 > 0 $, or  $ h'(x) > 0 $.
				\item 
				For $ x=0 $, we have $ h(0)= ln_{r}(0)+1 = -\frac{1}{r}+1 \geq 0 $ for $ r \geq 1 $. 
			\end{enumerate}
			Therefore, $h$ is decreasing for $0 < x\leq 1$, increasing for $x>1$, and $ h(0) \geq 0$ for $ r \geq 1$.  Also, $h(1)=0$. This implies that  $h(x) \geq 0$ for $x\geq 0$ with $r\geq 1$.  
		\end{proof}

	\section{$r$-deformed $\alpha$-$z$-R\'enyi relative entropy and its properties} 
		
		We have defined the $r$-deformed $\alpha$-$z$-R\'enyi relative entropy $D_{\alpha,z}^{(r)}(\rho || \sigma)$ in \autoref{r-z}. In this section, we explain the characteristics of $D_{\alpha,z}^{(r)}(\rho || \sigma)$. We establish the non-negativity property of $D_{\alpha,z}^{(r)}(\rho || \sigma)$ in \autoref{t}, unitary invariance of $D_{\alpha,z}^{(r)}(\rho || \sigma)$ in \autoref{h}, pseudo-additivity of $D_{\alpha,z}^{(r)}(\rho || \sigma)$ in \autoref{p}, a generalized data-processing inequality for $D_{\alpha,z}^{(r)}(\rho || \sigma)$ in \autoref{q}, and the joint convexity property of $D_{\alpha,z}^{(r)}(\rho || \sigma)$ in \autoref{r}.
		\begin{lemma} 
			For $r \neq 0$ and $\alpha \to 1$ with $\rho$ as density operator, $D_{\alpha,z}^{(r)}$ is reduced to the quantum relative entropy $S(\rho||\sigma)$, which is defined by $S(\rho||\sigma) =\tr(\rho\log\rho)+\tr(\rho \log\sigma)$.
		\end{lemma}
		 
		\begin{proof}
			This proof follows the L'H\^opital's Rule. The detailed calculation is as follows: 
			\begin{equation}
				\begin{split}
					\underset{\alpha\to1}{\lim}D_{\alpha,z}^{(r)}(\rho||\sigma)
					= & 	\underset{\alpha\to1}{\lim}\frac{1}{\alpha-1}\ln_{r}\left(\tr \left( \sigma^{\frac{1-\alpha}{2z}} \rho^{\frac{\alpha}{z}} \sigma^{\frac{1-\alpha}{2z}} \right)^z \right)
					\\
					= & 	\underset{\alpha\to1}{\lim}\frac{\left(\tr \left( \sigma^{\frac{1-\alpha}{2z}} \rho^{\frac{\alpha}{z}} \sigma^{\frac{1-\alpha}{2z}} \right)^z \right)^{r}-1}{r(\alpha-1)}\quad\left[ \frac{0}{0}\text{form}\right] 
					\\
					= & 	\underset{\alpha\to1}{\lim}\frac{{\frac{d}{d\alpha}}\left(\tr \left( \sigma^{\frac{1-\alpha}{2z}} \rho^{\frac{\alpha}{z}} \sigma^{\frac{1-\alpha}{2z}} \right)^z \right)^r-0}{\frac{d}{d\alpha}[r(\alpha-1)]}
					\\
					= & 	\underset{\alpha\to1}{\lim}\frac{r\left(\tr \left( \sigma^{\frac{1-\alpha}{2z}} \rho^{\frac{\alpha}{z}} \sigma^{\frac{1-\alpha}{2z}} \right)^z \right)^{r-1}\frac{d}{d\alpha}\left(\tr \left( \sigma^{\frac{1-\alpha}{2z}} \rho^{\frac{\alpha}{z}} \sigma^{\frac{1-\alpha}{2z}} \right)^z \right)}{r\frac{d}{d\alpha}(\alpha-1)}
					\\
					= & \underset{\alpha\to1}{\lim}\left(\tr \left( \sigma^{\frac{1-\alpha}{2z}} \rho^{\frac{\alpha}{z}} \sigma^{\frac{1-\alpha}{2z}} \right)^z \right)^{r-1}\tr\left(\frac{d}{d\alpha} \left( \sigma^{\frac{1-\alpha}{2z}} \rho^{\frac{\alpha}{z}} \sigma^{\frac{1-\alpha}{2z}} \right)^z \right)
					\\
					= & \underset{\alpha\to1}{\lim}\left(\tr \left( \sigma^{\frac{1-\alpha}{2z}} \rho^{\frac{\alpha}{z}} \sigma^{\frac{1-\alpha}{2z}} \right)^z \right)^{r-1}\tr\left[ z \left( \sigma^{\frac{1-\alpha}{2z}} \rho^{\frac{\alpha}{z}} \sigma^{\frac{1-\alpha}{2z}} \right)^{z-1}\frac{d}{d\alpha}\left( \sigma^{\frac{1-\alpha}{2z}} \rho^{\frac{\alpha}{z}} \sigma^{\frac{1-\alpha}{2z}} \right)\right]
					\\
					= & \underset{\alpha \to 	1}{\lim}\left(\tr \left( \sigma^{\frac{1-\alpha}{2z}} \rho^{\frac{\alpha}{z}} \sigma^{\frac{1-\alpha}{2z}} \right)^z \right)^{r-1} \\
					&  ~~~\tr\left[ z \left( 	\sigma^{\frac{1-\alpha}{2z}} \rho^{\frac{\alpha}{z}} \sigma^{\frac{1-\alpha}{2z}} \right)^{z-1}\left( \sigma^{\frac{1-\alpha}{2z}} \rho^{\frac{\alpha}{z}} \sigma^{\frac{1-\alpha}{2z}} \right)\left(-\frac{\log\sigma}{2z}+\frac{\log\rho}{z}-\frac{\log\sigma}{2z} \right) \right] 
					\\
					= & \underset{\alpha\to1}{\lim}\left(\tr \left( \sigma^{\frac{1-\alpha}{2z}} \rho^{\frac{\alpha}{z}} \sigma^{\frac{1-\alpha}{2z}} \right)^z \right)^{r-1}\tr\left[\left( \sigma^{\frac{1-\alpha}{2z}} \rho^{\frac{\alpha}{z}} \sigma^{\frac{1-\alpha}{2z}} \right)^z\left( \log\rho-\log\sigma\right) \right]
					\\
					= & 1^{r-1}\tr\left[ \rho\left( 	\log\rho-\log\sigma\right) \right]~~~~~~ [\text{since}\tr\rho=1]
					\\
					= & \tr{(\rho\log\rho-\rho\log\sigma)}.
				\end{split}
			\end{equation}
			Hence, the proof.
		\end{proof}
	
		Now we prove the non-negativity property of the $r$-deformed $\alpha$-$z$-R\'enyi relative entropy. Here, we apply the H\"{o}lder's inequality, which we mention below:
		\begin{lemma}
			\textbf{H\"{o}lder's inequality} \cite{albuquerque2017holder}: Let $(S, \Sigma, \mu)$ be a measure space and let $p, q \in [1, \infty]$ with $1/p + 1/q = 1$. Then for all measurable real or complex-valued functions $f$ and $g$ on $S$, 
			\begin{equation}
				\|fg\|_1 \leq \|f\|_p \|g\|_q.
			\end{equation}
		\end{lemma}
		The H\"{o}lder's inequality can be generalized for the operators between the Hilbert spaces as follows:
		\begin{lemma}\label{Operator_Holder_Inequality} 
			Let the operators $ A \in \mathcal{L}(\mathcal{H}_2, \mathcal{H}_3)$, and $B \in \mathcal{L}(\mathcal{H}_1, \mathcal{H}_2)$ be defined between the Hilbert spaces $\mathcal{H}_1, \mathcal{H}_2,$  and $ \mathcal{H}_3 $ respectively. Then, for $ p, q, z \in [0, \infty]$ satisfying $\frac{1}{p} + \frac{1}{q} = \frac{1}{z}$, we have
			\begin{equation}\label{Holders inequality}
				\| A B \|_z \leq \| A \|_p \| B \|_q.
			\end{equation}
			Here, $\| A \|_p = \left(\tr(| A |^p) \right)^{1/p}$ where $| A |= \sqrt{\left({A}^\dagger A \right)}$ is the Schatten-$p$ norm.
		\end{lemma}
		
		\begin{corollary}\label{corollary 1}
			For any density operator $\rho$ and a positive semi-definite operator $\sigma$ with $0 \leq \tr(\sigma) \leq 1$, $0 \leq \alpha < 1 $ and $ z \geq max \{\alpha, 1-\alpha\} $, we have $ \tr \left( \sigma^{\frac{1-\alpha}{2z}} \rho^{\frac{\alpha}{z}} \sigma^{\frac{1-\alpha}{2z}} \right)^z \leq 1 $ . 
		\end{corollary}
		\begin{proof}
			Put $p = \frac{z}{\alpha}$, $q = \frac{z}{1- \alpha}$, $A = \rho^{\frac{\alpha}{z}}$ and $B = \sigma^{\frac{1-\alpha}{z}}$ in \autoref{Operator_Holder_Inequality}. Now, for $0 \leq \alpha < 1 $ and $ z \geq max \{\alpha, 1-\alpha\} $, equation (\ref{Holders inequality}) suggests 
			\begin{equation*}
				\begin{split}
					&\|\rho^{\frac{\alpha}{z}}\sigma^{\frac{1-\alpha}{z}}\|_z \leq \|\rho^{\frac{\alpha}{z}}\|_{\frac{z}{\alpha}}\|\sigma^{\frac{1-\alpha}{z}}\|_{\frac{z}{1-\alpha}},\\
					\text{or}~
					& \left[ \tr\left(  \rho^{\frac{\alpha}{z}}\sigma^{\frac{1-\alpha}{z}}\right)^z \right]^{\frac{1}{z}} \leq \left[ \tr\left(\rho^\frac{\alpha}{z} \right)^\frac{z}{\alpha}  \right]^{\frac{\alpha}{z}} \left[ \tr\left( \sigma^\frac{1-\alpha}{z}\right)^\frac{z}{1-\alpha} \right]^{\frac{1-\alpha}{z}},\\
					\text{or}~
					& \left[ \tr\left(  \rho^{\frac{\alpha}{z}}\sigma^{\frac{1-\alpha}{z}}\right)^z \right]^{\frac{1}{z}} \leq \left[ \tr\left(\rho \right) \right]^{\frac{\alpha}{z}} \left[ \tr\left( \sigma \right) \right]^{\frac{1-\alpha}{z}},\\
					\text{or}~
					& \left[ \tr\left(  \rho^{\frac{\alpha}{z}}\sigma^{\frac{1-\alpha}{z}}\right)^z \right]^{\frac{1}{z}} \leq 1,\\
					\text{or}~
					&\tr\left(\rho^{\frac{\alpha}{z}}\sigma^{\frac{1-\alpha}{z}}\right)^z \leq 1,\\
					\text{or}~
					& \tr \left( \sigma^{\frac{1-\alpha}{2z}} \rho^{\frac{\alpha}{z}} \sigma^{\frac{1-\alpha}{2z}} \right)^z \leq 1.
				\end{split}
			\end{equation*}
			Hence, the proof. 
		\end{proof}
		
		\begin{lemma}\label{t}
			\textbf{Non-negativity:} For any density operator $\rho$ and a positive semi-definite operator $\sigma$ with $0 \leq \tr(\sigma) \leq 1$, the $r$-deformed $\alpha$-$z$-R\'enyi relative entropy $D_{\alpha,z}^{(r)}(\rho || \sigma) \geq 0 $. 
		\end{lemma}
		\begin{proof}
			Since, $\rho$ and $\sigma$ are positive semi definite operators, $\tr \left( \sigma^{\frac{1-\alpha}{2z}} \rho^{\frac{\alpha}{z}} \sigma^{\frac{1-\alpha}{2z}} \right)^z \geq 0$. Using \autoref{corollary 1}, for $\alpha\in[0,1)$ and max $ \{\alpha, 1-\alpha \}\leq z$ we have $0 \leq \tr \left( \sigma^{\frac{1-\alpha}{2z}} \rho^{\frac{\alpha}{z}} \sigma^{\frac{1-\alpha}{2z}} \right)^z \leq 1 $. From \autoref{properties of r-logagithm 3}, ~$\ln_{r}(x)$ is a monotone increasing function with $\ln_{r}(1)=0$, that is, $\ln_{r}(x)\leq 0$, for $x\in[0,1]$. Hence,
			\begin{equation}{\label{tk}}
				\begin{split}
					&\ln_{r}\left(\tr \left( 	\sigma^{\frac{1-\alpha}{2z}} \rho^{\frac{\alpha}{z}} \sigma^{\frac{1-\alpha}{2z}} \right)^z \right) \leq 0,\\
					\text{or }
					& \frac{1}{\alpha-1}\ln_{r}\left(\tr 	\left( \sigma^{\frac{1-\alpha}{2z}} \rho^{\frac{\alpha}{z}} \sigma^{\frac{1-\alpha}{2z}} \right)^z \right) \geq 0  ~~[\text{since~} 0 \leq \alpha<1] ,\\
					\text{or }
					& 	D_{\alpha,z}^{(r)}(\rho || \sigma) 	\geq 0 .
				\end{split}
			\end{equation}
			Hence, the proof.
		\end{proof} 
		
		\begin{lemma}
			\textbf{Order axiom:}
			Let $\rho$ be a density operator and $\sigma$ be a positive semi-definite operator. Then, for $|1-\alpha| \leq z$, we have  
			$$ D_{\alpha , z}^{(r)}(\rho || \sigma) \geq 0 \text{~when~} \rho \geq \sigma; \text{~and~} D_{\alpha , z}^{(r)}(\rho || \sigma) \leq 0 \text{~when~} \rho \leq \sigma. $$
		\end{lemma}
		
		\begin{proof}
			Due to $|1-\alpha| \leq z$, we consider the following two cases for $-z < (1 - \alpha)$ and $(1 - \alpha) \geq z$.  
			\begin{enumerate}[label=\textbf{Case \arabic*:}, leftmargin=*]
				\item 
				For $ \alpha < 1$, $(1-\alpha) \leq z$ and $\rho \geq \sigma $ using operator monotonicity, we have
				\begin{equation}
					\begin{split}
						&\tr\left(\rho^\frac{\alpha}{z} \sigma^\frac{1-\alpha}{z} \right)^z \leq  \tr\left(\rho^\frac{\alpha}{z} \rho^\frac{1-\alpha}{z} \right)^z, \\
						\text{or}~
						&\tr\left(\sigma^\frac{1-\alpha}{2z}\rho^\frac{\alpha}{z}\sigma^\frac{1-\alpha}{2z}\right)^z \leq \tr\left(\rho^\frac{1-\alpha}{2z}\rho^\frac{\alpha}{z}\rho^\frac{1-\alpha}{2z}\right)^z, \\
						\text{or}~
						&\ln_{r}\left( \tr\left(\sigma^\frac{1-\alpha}{2z}\rho^\frac{\alpha}{z}\sigma^\frac{1-\alpha}{2z}\right)^z \right) \leq \ln_{r}\left( \tr\left(\rho^\frac{1-\alpha}{2z}\rho^\frac{\alpha}{z}\rho^\frac{1-\alpha}{2z}\right)^z \right)~\text{[using \autoref{Properties of r-logarithm 2}]}\\
						\text{or}~
						&\ln_{r}\left( \tr\left(\sigma^\frac{1-\alpha}{2z}\rho^\frac{\alpha}{z}\sigma^\frac{1-\alpha}{2z}\right)^z \right) \leq 0 ~\text{[since $ \tr(\rho)=1 $]},\\
						\text{or}~
						&\frac{1}{\alpha-1} \ln_{r}\left( \tr\left(\sigma^\frac{1-\alpha}{2z}\rho^\frac{\alpha}{z}\sigma^\frac{1-\alpha}{2z}\right)^z \right) \geq 0 \\
						\text{or}~
						&D_{\alpha , z}^{(r)}(\rho || \sigma) \geq 0.
					\end{split}
				\end{equation}
				Using the same process we will get $D_{\alpha , z}^{(r)}(\rho || \sigma) \leq 0$ for $\rho \leq \sigma $.
				\item 
				For $\alpha > 1$, $ z> (\alpha-1)$ and $\rho \geq \sigma $, we have 
				\begin{equation}
					\rho^\frac{\alpha-1}{z} \geq \sigma^\frac{\alpha-1}{z} \text{~that is~}  \sigma^\frac{1-\alpha}{z} \geq \rho^\frac{1-\alpha}{z}.
				\end{equation}
				Therefore,
				\begin{equation}
					\begin{split}
						&\tr\left(\rho^\frac{\alpha}{z} \sigma^\frac{1-\alpha}{z} \right)^z \geq  \tr\left(\rho^\frac{\alpha}{z} \rho^\frac{1-\alpha}{z} \right)^z, \\
						\text{or}~
						&\tr\left(\sigma^\frac{1-\alpha}{2z}\rho^\frac{\alpha}{z}\sigma^\frac{1-\alpha}{2z}\right)^z \geq \tr\left(\rho^\frac{1-\alpha}{2z}\rho^\frac{\alpha}{z}\rho^\frac{1-\alpha}{2z}\right)^z, \\
						\text{or}~
						&\ln_{r}\left( \tr\left(\sigma^\frac{1-\alpha}{2z}\rho^\frac{\alpha}{z}\sigma^\frac{1-\alpha}{2z}\right)^z \right) \geq \ln_{r}\left( \tr\left(\rho^\frac{1-\alpha}{2z}\rho^\frac{\alpha}{z}\rho^\frac{1-\alpha}{2z}\right)^z \right)~[\text{using \autoref{Properties of r-logarithm 2}}]\\
						\text{or}~
						&\ln_{r}\left( \tr\left(\sigma^\frac{1-\alpha}{2z}\rho^\frac{\alpha}{z}\sigma^\frac{1-\alpha}{2z}\right)^z \right) \geq 0 ~[\text{since $ \tr(\rho)=1 $}],\\
						\text{or}~
						&\frac{1}{\alpha-1} \ln_{r}\left( \tr\left(\sigma^\frac{1-\alpha}{2z}\rho^\frac{\alpha}{z}\sigma^\frac{1-\alpha}{2z}\right)^z \right) \geq 0 \\
						\text{or}~
						&D_{\alpha , z}^{(r)}(\rho || \sigma) \geq 0.
					\end{split}
				\end{equation}
				Similarly, for $\rho \leq \sigma $ we get $ D_{\alpha , z}^{(r)}(\rho || \sigma) \leq 0$.
			\end{enumerate}
			Combining we get the proof.
		\end{proof}

	\begin{lemma}\label{h}
		\textbf{Unitary invariance:} For any unitary operator $U$, we have $D_{\alpha,z}^{(r)}\left( U \rho U^{\dagger} || U \sigma U^{\dagger}\right)  = D_{\alpha,z}^{(r)}\left( \rho || \sigma\right) $.
	\end{lemma}
	\begin{proof}
		As $\rho$ and $\sigma$ are positive  semi-definite operators, they are diagonalizable. Therefore, $(U \sigma U^{\dagger})^k = \left( U \sigma^k U^{\dagger}\right) $. Applying this relation, we get
		\begin{equation}\label{h2}
			\begin{split}
				D_{\alpha,z}^{(r)}(U \rho U^{\dagger} || U \sigma U^{\dagger}) 
				&=\frac{1}{\alpha - 1}\ln_{r}\tr \left[ \left( (U \sigma U^{\dagger})^{\frac{1 - \alpha}{2z}} (U \rho U^{\dagger})^{\frac{\alpha}{z}} (U \sigma U^{\dagger})^{\frac{1 - \alpha}{2z}} \right)^z \right]
				\\
				&=\frac{1}{\alpha - 1}\ln_{r}\tr \left[
				\left( U \sigma^{\frac{1 - \alpha}{2z}} U^\dagger \cdot U \rho^{\frac{\alpha}{z}} U^\dagger \cdot U \sigma^{\frac{1 - \alpha}{2z}} U^\dagger \right)^z \right] 
				\\
				&=\frac{1}{\alpha - 1}\ln_{r}\tr \left[ \left( U \left( \sigma^{\frac{1 - \alpha}{2z}} \rho^{\frac{\alpha}{z}} \sigma^{\frac{1 - \alpha}{2z}} \right) U^\dagger \right)^z \right]~~~\left[\text{since}~U U^\dagger = U^\dagger U = I \right]  
				\\
				&= \frac{1}{\alpha - 1}\ln_{r}\tr \left[ \left( \sigma^{\frac{1 - \alpha}{2z}} \rho^{\frac{\alpha}{z}} \sigma^{\frac{1 - \alpha}{2z}} \right)^z\right] =D_{\alpha,z}^{(r)}(\rho || \sigma).
			\end{split}
		\end{equation}
		Hence, the result.
	\end{proof}
	
	\begin{lemma}\label{h4}
		For any positive real number p and two positive semi-definite operators A, B we have $(A\otimes B)^p=A^p\otimes B^p $.
	\end{lemma}
	\begin{proof}
		For the operators $A\in M_{m,n}(F),B\in M_{q, r}(F),C\in M_{n,k}(F),D\in M_{r, s}(F)$, we have \cite{horn1994topics}
		\begin{equation}\label{ha}
			(A\otimes B)(C\otimes D)=(AC)\otimes(BD).
		\end{equation}
		Every positive semi-definite operator is Hermitian. Thus, they are diagonalizable. As $A$ and $B$ are two positive semi-definite operators, then they can be expressed as
		$A = \sum_{i=1}^n \lambda_i \ket{\psi_i} \bra{\psi_i}$  and $B = \sum_{j = 1}^n \mu_j \ket{\phi_j} \bra{\phi_j}$,
		where $\lambda_i$ and $\mu_i$ are the eigenvalues of $A$ and $B$, respectively, as well as $\ket{\psi_i}$ and $\ket{\phi_i}$ are the eigenvectors corresponding to $\lambda_i$ and $\mu_i$, respectively. Now,
		\begin{equation}
			\left( A\otimes B\right)^p
			=\sum_{i=1}^n\sum_{j=1}^n\lambda_{i}^p\mu_{j}^p \ket{\psi_i\otimes\phi_j} \bra{\psi_i\otimes\phi_j}
			=\left( \sum_{i=1}^n\lambda_{i}^p\ket{\psi_i} \bra{\psi_i}\right)\otimes\left(\sum_{j=1}^n\mu_{j}^p \ket{\phi_j} \bra{\phi_j}\right) 
			=A^p\otimes B^p.
		\end{equation}
		Hence, the proof.
	\end{proof}
	This Lemma we use to prove the pseudo-additivity property of the $r$-deformed $\alpha$-$z$ R\'{e}nyi relative entropy, which is as follows:
	
	\begin{lemma}\label{p}
		\textbf{Pseudo-additivity:} $D_{\alpha,z}^{(r)}(\rho_{1}\otimes\rho_{2}||\sigma_{1}\otimes\sigma_{2}) = D_{\alpha,z}^{(r)}(\rho_{1}||\sigma_{1})+D_{\alpha,z}^{(r)}(\rho_{2}||\sigma_{2})+r(\alpha-1)D_{\alpha,z}^{(r)}(\rho_{1}||\sigma_{1})D_{\alpha,z}^{(r)}(\rho_{2}||\sigma_{2}).$
	\end{lemma}
	\begin{proof}
		Expanding $D_{\alpha,z}^{(r)}(\rho_{1}\otimes\rho_{2}||\sigma_{1}\otimes\sigma_{2})$, we observe that
		\begin{equation}
			\begin{split}
				D_{\alpha,z}^{(r)}(\rho_{1}\otimes\rho_{2}||\sigma_{1}\otimes\sigma_{2})
				&=\frac{1}{\alpha - 1}\ln_{r}\tr \left[ \left( ( \sigma_{1}\otimes\sigma_{2})^{\frac{1 - \alpha}{2z}} (\rho_{1}\otimes\rho_{2})^{\frac{\alpha}{z}} (\sigma_{1}\otimes\sigma_{2})^{\frac{1 - \alpha}{2z}} \right)^z \right]\\
				&=\frac{1}{\alpha-1}\ln_{r}\tr\left[
				\left( \sigma_1^{\frac{1 - \alpha}{2z}} \otimes \sigma_2^{\frac{1 - \alpha}{2z}} \right)
				\left( \rho_1^{\frac{\alpha}{z}} \otimes \rho_2^{\frac{\alpha}{z}} \right)
				\left( \sigma_1^{\frac{1 - \alpha}{2z}} \otimes \sigma_2^{\frac{1 - \alpha}{2z}} \right)
				\right]^z,~~[\text{using \autoref{h4}}]\\
				&=\frac{1}{\alpha-1}\ln_{r}\tr\left[
				\left( \sigma_1^{\frac{1 - \alpha}{2z}} \rho_1^{\frac{\alpha}{z}} \sigma_1^{\frac{1 - \alpha}{2z}} \right)
				\otimes
				\left( \sigma_2^{\frac{1 - \alpha}{2z}} \rho_2^{\frac{\alpha}{z}} \sigma_2^{\frac{1 - \alpha}{2z}} \right)
				\right]^z,~~~~[\text{using equation (\ref{ha})}]\\
				&=\frac{1}{\alpha-1}\ln_{r}\tr\left[ \left( \sigma_1^{\frac{1 - \alpha}{2z}} \rho_1^{\frac{\alpha}{z}} \sigma_1^{\frac{1 - \alpha}{2z}} \right)^z
				\otimes\left( \sigma_2^{\frac{1 - \alpha}{2z}} \rho_2^{\frac{\alpha}{z}} \sigma_2^{\frac{1 - \alpha}{2z}} \right)^z\right],~~[\text{using \autoref{h4}}]\\
				&=\frac{1}{\alpha-1}\ln_{r}\left[\tr \left( \sigma_1^{\frac{1 - \alpha}{2z}} \rho_1^{\frac{\alpha}{z}} \sigma_1^{\frac{1 - \alpha}{2z}} \right)^z
				\cdot\tr
				\left( \sigma_2^{\frac{1 - \alpha}{2z}} \rho_2^{\frac{\alpha}{z}} \sigma_2^{\frac{1 - \alpha}{2z}} \right)^z\right], 
			\end{split}
		\end{equation}
		since $\tr(A\otimes B)=\tr(A)\tr(B)$. Using \autoref{Properties of r-logarithm 1}, we have $\ln_{r}(xy)=\ln_{r}(x)+\ln_{r}(y)+r\ln_{r}(x)\ln_{r}(x)$. Therefore,
		\begin{equation}
			\begin{split}
				D_{\alpha,z}^{(r)}(\rho_{1}\otimes\rho_{2}||\sigma_{1}\otimes\sigma_{2})
				&= \frac{1}{\alpha-1}\ln_{r}\tr\left( \sigma_1^{\frac{1 - \alpha}{2z}} \rho_1^{\frac{\alpha}{z}} \sigma_1^{\frac{1 - \alpha}{2z}} \right)^z+\frac{1}{\alpha-1}\ln_{r}\tr\left( \sigma_2^{\frac{1 - \alpha}{2z}} \rho_2^{\frac{\alpha}{z}} \sigma_2^{\frac{1 - \alpha}{2z}} \right)^z + 
				\\
				& \hspace{1cm}  r\frac{1}{\alpha-1}\ln_{r}\tr\left( \sigma_1^{\frac{1 - \alpha}{2z}} \rho_1^{\frac{\alpha}{z}} \sigma_1^{\frac{1 - \alpha}{2z}} \right)^z\ln_{r}\tr\left( \sigma_2^{\frac{1 - \alpha}{2z}} \rho_2^{\frac{\alpha}{z}} \sigma_2^{\frac{1 - \alpha}{2z}} \right)^z
				\\
				&=D_{\alpha,z}^{(r)}(\rho_{1}||\sigma_{1})+D_{\alpha,z}^{(r)}(\rho_{2}||\sigma_{2})+r(\alpha-1)D_{\alpha,z}^{(r)}(\rho_{1}||\sigma_{1})D_{\alpha,z}^{(r)}(\rho_{2}||\sigma_{2}).
			\end{split}
		\end{equation}
		Hence, the proof.
	\end{proof}
	
	Now, we describe the convexity property of $D_{\alpha, z}^{(r)}$. The operator concave and convex functions \cite{bhatia2013matrix} are defined as follows.
	\begin{definition}
		Let $f(A,B)$ be a real valued function of two operator variables and $ 0 \leq p \leq 1$. We say $f$ is jointly concave, if
		\begin{equation}
			f(pA_{1}+(1-p)A_{2}, pB_{1}+(1-p)B_{2}) \geq pf(A_{1},B_{1})+(1-p)f(A_{2},B_{2});
		\end{equation}
		 and jointly convex, if
		\begin{equation}
			f(pA_{1}+(1-p)A_{2}, pB_{1}+(1-p)B_{2}) \leq pf(A_{1},B_{1})+(1-p)f(A_{2},B_{2});
		\end{equation}
		for all $A_{1}$, $A_{2}$, $B_{1}$, and $B_{2}$.
	\end{definition}
	
	\begin{lemma}\label{r}
		\textbf{Joint convexity:} For $p\in [0,1]$, 
		we have
		$$ D_{\alpha,z}^{(r)}\left(p\rho_{1}+(1-p)\rho_{2}||p\sigma_{1}+(1-p)\sigma_{2} \right)\leq pD_{\alpha,z}^{(r)}(\rho_{1}||\sigma_{1})+(1-p)	D_{\alpha,z}^{(r)}(\rho_{2}||\sigma_{2}).$$
	\end{lemma}
	\begin{proof}
		Consider the trace functional
		\begin{equation}\label{trace_functional_1}
			f_{\alpha,z}(\rho,\sigma)= \tr\left( \sigma^{\frac{1-\alpha}{2z}} \rho^{\frac{\alpha}{z}} \sigma^{\frac{1-\alpha}{2z}} \right)^z.
		\end{equation}
		Since, $f_{\alpha,z}(\rho || \sigma)$ is jointly concave function for $0 < \alpha \leq 1$ and max $\{\alpha, 1-\alpha \} \leq z$ \cite{hiai2013concavity} therefore,
		\begin{equation*}
			\begin{split}
				& f_{\alpha,z}(p\rho_{1}+(1-p)\rho_{2},p\sigma_{1}+(1-p)\sigma_{2})\geq pf_{\alpha,z}(\rho_{1},\sigma_{1})+ (1-p)f_{\alpha,z}(\rho_{2},\sigma_{2}) \\
				\text{or}~ 
				&\tr\left( (p\sigma_{1}+(1-p)\sigma_{2})^\frac{1-\alpha}{2z}(p\rho_{1}+(1-p)\rho_{2})^\frac{\alpha}{z}(p\sigma_{1}+(1-p)\sigma_{2})^\frac{1-\alpha}{2z}\right)^{z}  \\ & \hspace{4cm} \geq
				p\tr\left( \sigma_{1}^\frac{1-\alpha}{2z}\rho_{1}^\frac{\alpha}{z}\sigma_{1}^\frac{1-\alpha}{2z}\right)^{z} +(1-p)\tr\left( \sigma_{2}^\frac{1-\alpha}{2z}\rho_{2}^\frac{\alpha}{z}\sigma_{2}^\frac{1-\alpha}{2z}\right)^{z} \\
				\text{or}~
				&\ln_{r}\left(\tr\left( (p\sigma_{1}+(1-p)\sigma_{2})^\frac{1-\alpha}{2z}(p\rho_{1}+(1-p)\rho_{2})^\frac{\alpha}{z}(p\sigma_{1}+(1-p)\sigma_{2})^\frac{1-\alpha}{2z}\right)^{z}\right)  \\ & \hspace{4cm} \geq
				\ln_{r}\left[p\tr\left(\sigma_{1}^\frac{1-\alpha}{2z}\rho_{1}^\frac{\alpha}{z}\sigma_{1}^\frac{1-\alpha}{2z}\right)^{z}+(1-p)\tr\left( \sigma_{2}^\frac{1-\alpha}{2z}\rho_{2}^\frac{\alpha}{z}\sigma_{2}^\frac{1-\alpha}{2z}\right)^{z}\right], ~ \text{using \autoref{properties of r-logagithm 3}} \\
				\text{or}~
				&\ln_{r}\left(\tr\left( (p\sigma_{1}+(1-p)\sigma_{2})^\frac{1-\alpha}{2z}(p\rho_{1}+(1-p)\rho_{2})^\frac{\alpha}{z}(p\sigma_{1}+(1-p)\sigma_{2})^\frac{1-\alpha}{2z}\right)^{z}\right) \\ & \hspace{4cm}
				\geq  p\ln_{r}\tr\left(\sigma_{1}^\frac{1-\alpha}{2z}\rho_{1}^\frac{\alpha}{z}\sigma_{1}^\frac{1-\alpha}{2z}\right)^{z}+(1-p)\ln_{r}\tr\left( \sigma_{2}^\frac{1-\alpha}{2z}\rho_{2}^\frac{\alpha}{z}\sigma_{2}^\frac{1-\alpha}{2z}\right)^{z},~\text{using \autoref{Properties of r-logarithm 2}}\\
				\text{or}~
				&\frac{1}{\alpha-1}\ln_{r}\left(\tr\left( (p\sigma_{1}+(1-p)\sigma_{2})^\frac{1-\alpha}{2z}(p\rho_{1}+(1-p)\rho_{2})^\frac{\alpha}{z}(p\sigma_{1}+(1-p)\sigma_{2})^\frac{1-\alpha}{2z}\right)^{z}\right) \\ & \hspace{4cm} \leq  \frac{1}{\alpha-1}p\ln_{r}\tr\left(\sigma_{1}^\frac{1-\alpha}{2z}\rho_{1}^\frac{\alpha}{z}\sigma_{1}^\frac{1-\alpha}{2z}\right)^{z}+(1-p)\ln_{r}\tr\left( \sigma_{2}^\frac{1-\alpha}{2z}\rho_{2}^\frac{\alpha}{z}\sigma_{2}^\frac{1-\alpha}{2z}\right)^{z}\\
				\text{or}~
				& D_{\alpha,z}^{(r)}\left(p\rho_{1}+(1-p)\rho_{2}||p\sigma_{1}+(1-p)\sigma_{2} \right)\leq pD_{\alpha,z}^{(r)}(\rho_{1}||\sigma_{1})+(1-p)	D_{\alpha,z}^{(r)}(\rho_{2}||\sigma_{2}).
			\end{split}
		\end{equation*}
		
		Thus, by Lieb convexity theorem \cite{bhatia2013matrix}, $D_{\alpha}^{(r)}(\rho||\sigma)$ is jointly convex for $0 < \alpha \leq 1$, max $\{\alpha, 1-\alpha \} \leq z$ and $r\in(0,1]$. Again $f_{\alpha,z}(\rho||\sigma)$ is jointly convex function for
		\begin{itemize}
			\item $1\leq \alpha \leq 2$ and $z=1$\cite{ando1979concavity},
			\item $1\leq \alpha$ and $z=\alpha$\cite{frank2013monotonicity,beigi2013sandwiched}, 
			\item $1 \leq \alpha \leq 2$ and $z=\frac{\alpha}{2}$\cite{carlen2016some}.
		\end{itemize}
		Taking $r>1$, we get the same result for these three cases.
	\end{proof}
	
	\begin{definition}
		\textbf{Completely Positive Trace-Preserving map (CPTP map)\cite{choi1975completely, ruskai2002analysis}:} Let $\mathcal{M}_n$ be the set of complex operators of order $n$ and $A$ be a positive semi-definite operator on $\mathcal{M}_n$. Then, linear map $\Phi$ from $\mathcal{M}_n$ to $\mathcal{M}_m$ is completely positive iff it admits an expression 
		$\Phi(A)=\sum_{i}{V_i}^\dagger A {V_i}$ where $V_i$ are operators of order $n \times m$. The map $\Phi$ is said to be trace-preserving if $\sum_{i}{V_i}{V_i}^\dagger=I$. 
	\end{definition}
	\begin{lemma}\label{n}
		Let $\rho$, $\sigma$ be two positive semi-definite operators and $\tau$ be a density operator then $f_{\alpha,z}(\rho,\sigma)=f_{\alpha,z}(\rho \otimes \tau,\sigma\otimes\tau).$
	\end{lemma}
	\begin{proof}
		We have $f_{\alpha,z}(\rho,\sigma)= \tr\left( \sigma^{\frac{1-\alpha}{2z}} \rho^{\frac{\alpha}{z}} \sigma^{\frac{1-\alpha}{2z}} \right)^z $. Therefore,
		\begin{equation}
			\begin{split}
				f_{\alpha,z}(\rho \otimes \tau,\sigma\otimes\tau) &= \tr\left( (\sigma \otimes \tau)^{\frac{1-\alpha}{2z}} (\rho\otimes \tau)^{\frac{\alpha}{z}} (\sigma\otimes \tau)^{\frac{1-\alpha}{2z}} \right)^z, \\
				&= \tr\left( (\sigma^{\frac{1-\alpha}{2z}} \otimes \tau^{\frac{1-\alpha}{2z}}) (\rho^{\frac{\alpha}{z}}\otimes \tau^{\frac{\alpha}{z}}) (\sigma^{\frac{1-\alpha}{2z}}\otimes \tau^{\frac{1-\alpha}{2z}}) \right)^z , ~~[\text{using \autoref{h4}}]\\
				&= \tr\left[ \left( \sigma^{\frac{1-\alpha}{2z}} \rho^{\frac{\alpha}{z}} \sigma^{\frac{1-\alpha}{2z}} \right) \otimes \left( \tau^{\frac{1-\alpha}{2z}} \tau^{\frac{\alpha}{z}} \tau^{\frac{1-\alpha}{2z}} \right) \right]^z, ~~[\text{using equation (\ref{ha})}]\\
				&= \tr\left[ \left( \sigma^{\frac{1-\alpha}{2z}} \rho^{\frac{\alpha}{z}} \sigma^{\frac{1-\alpha}{2z}} \right)^z \otimes \left( \tau^{\frac{1-\alpha}{2z}} \tau^{\frac{\alpha}{z}} \tau^{\frac{1-\alpha}{2z}} \right)^z \right],~~[\text{using \autoref{h4}}]\\
				&= \tr \left( \sigma^{\frac{1-\alpha}{2z}} \rho^{\frac{\alpha}{z}} \sigma^{\frac{1-\alpha}{2z}} \right)^z . \tr\left( \tau^{\frac{1}{z}} \right)^z, \text{since} \tr(A\otimes B)= \tr(A)\tr(B)\\
				&= f_{\alpha,z}(\rho,\sigma), ~\text{since} \tr(\tau)=1.
			\end{split}
		\end{equation}
		Hence, the proof.
	\end{proof}

	\begin{lemma}\label{u}
		Let $\mathcal{H}$ be a finite-dimensional Hilbert space with dimension $n$ and $u$ be a unitary operator in $\mathcal{H}$. Then
		\begin{equation}
			\int du \left( u\ket{i}\bra{j}u^{\dagger}\right) = \begin{cases} 
				\frac{I}{n}, & \text{if } i=j; \\ 
				0, & \text{otherwise}. 
			\end{cases}
		\end{equation}
		Here $du$ is normalized Haar measure on all unitaries defined on $\mathcal{H}$.
	\end{lemma}
	\begin{proof}
		Form \cite{barthel2023absence}, We have
		\begin{equation}\label{u1}
			\int du \left( u_{i,j}u_{m,n}^{\dagger}\right) =\frac{I}{n}\delta_{i,m}\delta_{j,n}.
		\end{equation}
		Now, $u\ket{i}=\sum_{k}u_{k,i}\ket{k}$ and $\bra{j}u^{\dagger}=\left( u\ket{j}\right)^\dagger=\sum_{l}u_{l,j}^{\dagger}\bra{l}$. Therefore,
		\begin{equation}
			\begin{split}
				\int du \left( u\ket{i}\bra{j}u^{\dagger}\right) =\sum_{k,l}\left( \int du (u_{k,i} u_{l,j}^{\dagger})\right)\ket{k}\bra{l}
				&=\sum_{k,l}\left( \frac{I}{n}\delta_{k,l}\delta_{i,j}\right)\ket{k}\bra{l},~~[\text{using equation (\ref{u1})}]\\
				&=\frac{\delta_{i,j}}{n}I \\
				&=\begin{cases} 
					\frac{I}{n}, & \text{if } i=j; \\ 
					0, & \text{otherwise}. 
				\end{cases}
			\end{split}
		\end{equation}
		Hence, the proof.
	\end{proof}
	
	\begin{lemma}\label{v}
		For any two positive semi-definite operators $\rho$, $\sigma$ and a jointly concave function $f_{\alpha,z}$ in $(\rho,\sigma)$, we have $f_{\alpha,z}(\rho,\sigma) \leq f_{\alpha,z}(\Phi(\rho),\Phi(\sigma))$, where $\Phi$ is a CPTP map.
	\end{lemma}
	\begin{proof}
		Let $\mathcal{H}_{1}$ and $\mathcal{H}_2$ be two finite-dimensional Hilbert spaces with dimension $n$ and $\rho$, $\sigma$ defined on $\mathcal{H}_1$. Then, for any CPTP map $\Phi$ there exists a pure state $\tau \in \mathcal{H}_2$ and a unitary operator $U^{(1,2)}$ defined on $(\mathcal{H}_1 \otimes \mathcal{H}_2)$ \cite{wilde2013quantum} such that $\Phi(\rho)=\tr_2 [U^{(1,2)}(\rho\otimes\tau)(U^{(1,2)})^{\dagger}]$. Let $U^{(1,2)}(\rho\otimes\tau)(U^{(1,2)})^{\dagger}= [A_{i,j}]$, then this can be written as:
		\begin{equation}
			U^{(1,2)}(\rho\otimes\tau)(U^{(1,2)})^{\dagger}=\sum_{i,j} A_{i,j}\otimes \ket{i}\bra{j}.
		\end{equation}
		Now, for unitary operator $u\in \mathcal{H}_{2}$, we have
		\begin{equation}
			(I \otimes u) U^{(1,2)}(\rho\otimes\tau)(U^{(1,2)})^{\dagger}(I \otimes u^{\dagger})= \sum_{i,j} A_{i,j}\otimes u\ket{i}\bra{j}u^{\dagger}.
		\end{equation}
		Let $du$ be normalized Haar measure on all unitaries on $\mathcal{H}_{2}$. Therefore,
		\begin{equation}\label{v1}
			\begin{split}
				\int du (I \otimes u) U^{(1,2)}(\rho\otimes\tau)(U^{(1,2)})^{\dagger}(I \otimes u^{\dagger})
				=&\int du \left( \sum_{i,j} A_{i,j}\otimes u\ket{i}\bra{j}u^{\dagger} \right)\\
				=&\sum_{i} A_{i,i}\otimes \frac{I}{n}, ~~~[\text{using \autoref{u}}] \\
				=&\tr_{2}\left[A_{i,j}\otimes \ket{i}\bra{j}\right]\otimes \frac{I}{n}\\
				=&\tr_{2}\left[U^{(1,2)}(\rho\otimes\tau)(U^{(1,2)})^{\dagger}\right]\otimes \frac{I}{n}\\
				=&\Phi(\rho)\otimes \frac{I}{n} .
			\end{split}
		\end{equation}
		Now,
		\begin{equation}
			\begin{split}
				&f_{\alpha,z}(\Phi(\rho),\Phi(\sigma))\\
				=&f_{\alpha,z}(\Phi(\rho)\otimes \frac{I}{n},\Phi(\sigma)\otimes \frac{I}{n}),~~[\text{using \autoref{n}}]\\
				=&f_{\alpha,z}\left( \int du (I \otimes u) U^{(1,2)}(\rho\otimes\tau)(U^{(1,2)})^{\dagger}(I \otimes u^{\dagger}),\int du (I \otimes u) U^{(1,2)}(\sigma\otimes\tau)(U^{(1,2)})^{\dagger}(I \otimes u^{\dagger})\right)~~[\text{using equation(\ref{v1}})]\\
				=&f_{\alpha,z}\left( \int du V_{u}(\rho\otimes\tau)(V_{u})^{\dagger},\int du V_{u}(\sigma\otimes\tau)(V_{u})^{\dagger}\right),~~[\text{putting~} (I \otimes u) U^{(1,2)}=V_{u}]\\
				\geq & \int du~f_{\alpha,z}\left( V_{u}(\rho\otimes\tau)(V_{u})^{\dagger}, V_{u}(\sigma\otimes\tau)(V_{u})^{\dagger}\right), [\text{since $f_{\alpha,z}$ is jointly concave}] \\
				=& \int du~f_{\alpha,z}\left( (\rho\otimes\tau), (\sigma\otimes\tau)\right), [\text{using equation (\ref{h2})}]\\
				=&f_{\alpha,z}\left( (\rho\otimes\tau), (\sigma\otimes\tau)\right),~~[\text{since $du$ is normalized}]\\
				=&f_{\alpha,z}(\rho,\sigma)~~[\text{using \autoref{n}}].
			\end{split}
		\end{equation}
		Hence, the proof.
	\end{proof}

		\begin{lemma}\label{q}
		\textbf{Data-processing inequality:} Let $\Phi$  be a CPTP map and $\rho$, $\sigma$ be two positive semi-definite operators, then $D_{\alpha,z}^{(r)}\left(\Phi(\rho)||\Phi(\sigma) \right)\leq D_{\alpha,z}^{(r)}(\rho||\sigma)$.
	\end{lemma}
	\begin{proof}
		We can rewrite the $r$-deformed $\alpha$-$z$ Rényi relative entropy as $D_{\alpha,z}^{(r)}(\rho || \sigma) = \frac{1}{\alpha - 1} \ln_{r} f_{\alpha,z}(\rho,\sigma)$, where the trace functional $f_{\alpha,z}(\rho,\sigma)$ is defined in equation (\ref{trace_functional_1}).
		 Now, concavity of $f_{\alpha,z}(\rho,\sigma)$ in $0 < \alpha \leq 1$, and max $\{\alpha, 1-\alpha \} \leq z$ follows from a concavity theorem proven by Hiai \cite{hiai2013concavity}. Therefore from \autoref{v},we have
		\begin{equation}
			\begin{split}
				&f_{\alpha,z}(\rho,\sigma)\leq f_{\alpha,z}(\Phi(\rho),\Phi(\sigma)),\\ 
				\text{or}~
				&\tr\left( \sigma^{\frac{1-\alpha}{2z}} \rho^{\frac{\alpha}{z}} \sigma^{\frac{1-\alpha}{2z}} \right)^z \leq
				\tr\left( \Phi(\sigma)^{\frac{1-\alpha}{2z}} \Phi(\sigma)^{\frac{\alpha}{z}} \Phi(\sigma)^{\frac{1-\alpha}{2z}} \right)^z,\\
				\text{or}~
				&\ln_{r}\left( \tr\left( \sigma^{\frac{1-\alpha}{2z}} \rho^{\frac{\alpha}{z}} \sigma^{\frac{1-\alpha}{2z}} \right)^z\right)  \leq
				\ln_{r}\left( \tr\left( \Phi(\sigma)^{\frac{1-\alpha}{2z}} \Phi(\sigma)^{\frac{\alpha}{z}} \Phi(\sigma)^{\frac{1-\alpha}{2z}} \right)^z\right),\\
				\text{or}~
				&\frac{1}{\alpha-1}\ln_{r}\left( \tr\left( \sigma^{\frac{1-\alpha}{2z}} \rho^{\frac{\alpha}{z}} \sigma^{\frac{1-\alpha}{2z}} \right)^z\right)  \geq\frac{1}{\alpha-1}\ln_{r}\left( \tr\left( \Phi(\sigma)^{\frac{1-\alpha}{2z}} \Phi(\sigma)^{\frac{\alpha}{z}} \Phi(\sigma)^{\frac{1-\alpha}{2z}} \right)^z\right),\\
				\text{or}~
				& D_{\alpha,z}^{(r)}\left(\Phi(\rho)||\Phi(\sigma) \right)\leq D_{\alpha,z}^{(r)}(\rho||\sigma).
			\end{split}
		\end{equation}
		Now, \autoref{v} gives reverse inequality for jointly convex function. Therefore, convexity of $f_{\alpha,z}(\rho|| \sigma)$, as discuss in the proof of \autoref{r}, follows similar result. 
		Hence, the proof.
	\end{proof}

	\section{Upper bounds of the Tsallis divergence and their comparison}
	
	\begin{lemma}
		\textbf{Araki–Lieb–Thirring inequality \cite{audenaert2007araki}:} Let $q\geq 0$. For any two positive semi-definite operators $A$, and $B$, we have,
		\begin{equation}\label{alti1}
			\tr((B^tA^tB^t)^q)\leq\tr((BAB)^{tq}),
		\end{equation}
		for $0\leq t \leq 1$. Also, for $t\geq 1$ we have,
		\begin{equation}\label{alti2}
			\tr((B^tA^tB^t)^q)\geq\tr((BAB)^{tq}).
		\end{equation}
	\end{lemma}

	\begin{theorem}\label{g}
			Let $\rho$ be a density operator and $\sigma$ be positive semi-definite operator defined on a Hilbert space $\mathcal{H}$ and $\alpha \in \mathbb{R} - \{1\}$. Then
		\begin{equation}
			D_{\alpha}(\rho || \sigma) \leq D_{\alpha,z}^{(r)}(\rho\|\sigma) ~~\text{for~}  r \in \mathbb{R} - \{0\}.
		\end{equation}
		Also, equality holds for $z=r=1$.
	\end{theorem}
	
	\begin{proof}
		We distribute this proof into two cases, which are as follows:
		\begin{enumerate}[label=\textbf{Case \arabic*:}, leftmargin=*]
			\item 
			Let $\alpha > 1$, $r \geq 1 \text{~and~} 0<z\leq 1.$ Also, let $t=\frac{1}{z}$ and $q=z$. Then from equation (\ref{alti2}), we have
			\begin{equation}\label{g2}
				\tr\left( \sigma^{\frac{1-\alpha}{2z}} \rho^{\frac{\alpha}{z}} \sigma^{\frac{1-\alpha}{2z}} \right)^z \geq\tr\left( \sigma^{\frac{1-\alpha}{2}}\rho^{\alpha}\sigma^{\frac{1-\alpha}{2}}\right) =\tr\left( \rho^{\alpha}\sigma^{1-\alpha}\right)~~\text{for~} 0 < z \leq 1. 
			\end{equation}
			\autoref{properties of r-logagithm 3} suggests that the $r$-logarithm is a monotone increasing function. Therefore, from equation (\ref{g2}) we say that
			\begin{equation}\label{g8}
				\ln_{r}\left[ \tr\left( \sigma^{\frac{1-\alpha}{2z}} \rho^{\frac{\alpha}{z}} \sigma^{\frac{1-\alpha}{2z}} \right)^z\right] \geq \ln_{r}\left[\tr\left( \rho^{\alpha}\sigma^{1-\alpha}\right)\right]. 
			\end{equation}
			
			Since, $\rho$ and $\sigma$ are positive semi-definite operators, therefore $\tr\left( \rho^{\alpha}\sigma^{1-\alpha}\right) \geq 0$. Now, applying \autoref{a} in equation (\ref{g8}) we get
			\begin{equation}\label{g4}
				\ln_{r}\left[ \tr\left( \sigma^{\frac{1-\alpha}{2z}} \rho^{\frac{\alpha}{z}} \sigma^{\frac{1-\alpha}{2z}} \right)^z\right] \geq \ln_{r}\left[\tr\left( \rho^{\alpha}\sigma^{1-\alpha}\right)\right]\geq \tr\left( \rho^{\alpha}\sigma^{1-\alpha}\right)-1, ~~\text{for~} r\geq 1.
			\end{equation}
			Since, $\rho$ is a density operator therefore $\tr(\rho)=1$. Using this in equation (\ref{g4}) we observe that
			\begin{equation}\label{g7}
				\begin{split}
					\ln_{r}\left[ \tr\left( \sigma^{\frac{1-\alpha}{2z}} \rho^{\frac{\alpha}{z}} \sigma^{\frac{1-\alpha}{2z}} \right)^z\right]&\geq \tr\left( \rho^{\alpha}\sigma^{1-\alpha}-\rho\right)\\
					\text{or~}
					\frac{1}{\alpha-1}\ln_{r}\left[ \tr\left( \sigma^{\frac{1-\alpha}{2z}} \rho^{\frac{\alpha}{z}} \sigma^{\frac{1-\alpha}{2z}} \right)^z\right] &\geq \tr\left(\frac{\rho^{\alpha}\sigma^{1-\alpha}-\rho}{\alpha-1} \right), ~~\text{since~} \alpha>1.
				\end{split}
			\end{equation}
			Therefore, $D_{\alpha}(\rho || \sigma) \leq D_{\alpha,z}^{(r)}(\rho\|\sigma)$ for $\alpha > 1$, $r \geq 1$ and $0 < z \leq 1$.
			
			\item 
			Let $\alpha < 1$, $r < 1$ and $z\geq 1.$ Also, let $\rho$ and $\sigma$ be positive semi-definite  operators. Putting $t=\frac{1}{z}$ and $q=z$ in equation (\ref{alti1}), we get
			\begin{equation}\label{j2}
				\tr\left( \sigma^{\frac{1-\alpha}{2z}} \rho^{\frac{\alpha}{z}} \sigma^{\frac{1-\alpha}{2z}} \right)^z \leq\tr\left( \sigma^{\frac{1-\alpha}{2}}\rho^{\alpha}\sigma^{\frac{1-\alpha}{2}}\right) =\tr\left( \rho^{\alpha}\sigma^{1-\alpha}\right)~~\text{for~}  z \geq 1. 
			\end{equation}
			From \autoref{properties of r-logagithm 3}, we say that the $r$-logarithm is a monotone increasing function. Hence, from equation (\ref{j2}) we have
			\begin{equation}\label{j7}
				\ln_{r}\left[ \tr\left( \sigma^{\frac{1-\alpha}{2z}} \rho^{\frac{\alpha}{z}} \sigma^{\frac{1-\alpha}{2z}} \right)^z\right] \leq \ln_{r}\left[\tr\left( \rho^{\alpha}\sigma^{1-\alpha}\right)\right].
			\end{equation}
			Since, $\rho$ and $\sigma$ are positive semi-definite operators, therefore $\tr\left( \rho^{\alpha}\sigma^{1-\alpha}\right) \geq 0$. Now, using \autoref{b} in equation (\ref{j7}), we get
			\begin{equation}\label{j4}
				\ln_{r}\left[ \tr\left( \sigma^{\frac{1-\alpha}{2z}} \rho^{\frac{\alpha}{z}} \sigma^{\frac{1-\alpha}{2z}} \right)^z\right] \leq \ln_{r}\left[\tr\left( \rho^{\alpha}\sigma^{1-\alpha}\right)\right]\leq \tr\left( \rho^{\alpha}\sigma^{1-\alpha}\right)-1 ~~\text{for~} r<1.
			\end{equation}
			Replacing $\tr(\rho)=1$  in equation (\ref{j4}), we have   
			\begin{equation}\label{j11}
				\begin{split}
					\ln_{r}\left[ \tr\left( \sigma^{\frac{1-\alpha}{2z}} \rho^{\frac{\alpha}{z}} \sigma^{\frac{1-\alpha}{2z}} \right)^z\right]&\leq \tr\left( \rho^{\alpha}\sigma^{1-\alpha}-\rho\right)\\
					\text{or~}
					\frac{1}{\alpha-1}\ln_{r}\left[ \tr\left( \sigma^{\frac{1-\alpha}{2z}} \rho^{\frac{\alpha}{z}} \sigma^{\frac{1-\alpha}{2z}} \right)^z\right] &\geq \tr\left(\frac{\rho^{\alpha}\sigma^{1-\alpha}-\rho}{\alpha-1} \right), ~~\text{since~} \alpha<1.
				\end{split}
			\end{equation}
			Therefore, $D_{\alpha}(\rho || \sigma) \leq D_{\alpha,z}^{(r)}(\rho\|\sigma)$ for $\alpha < 1$, $r < 1$ and $z\geq 1$.
		\end{enumerate}
			
		Combining \textbf{Case 1} and \textbf{Case 2} we have the proof. The equality follows from equation (\ref{equality_of_Tsalish_bound}).
	\end{proof}
	
	We derive the inequality (\ref{j4}) using the property of $r$-logarithm. A similar inequality is derived in \cite{fujii2021matrix} by utilizing the operator mejorization inequalities. It may be considered as a benefit of using the $r$-logarithm as the application of operator mejorization inequalities are not essential.
	
	To derive equation (\ref{j11}) from equation (\ref{j4}) we replace 1 by $\tr(\rho)$, as $\rho$ is a density operator. We can extend it further by replacing it with any positive semi-definite operator such that $0\leq \tr(\rho) \leq 1$.
	
	\autoref{g} shows that $D_{\alpha, z}^{(r)}(\rho || \sigma)$ is an upper bound of the Tsallis relative entropy $D_{\alpha}(\rho||\sigma)$. It is known \cite{fujii2021matrix, furuichi2010matrix, seo2019matrix} that the upper bound for $D_{\alpha}(\rho||\sigma)$ was also given by the Tsallis relative operator entropy as follows:
	\begin{equation}
		D_{\alpha}(\rho||\sigma) \leq -\tr \left[\rho \ln_{1-\alpha}\left(\rho^{-q/2}\sigma^q\rho^{-q/2}\right)^{1/q}\right],
	\end{equation}
	where $-1 \le \alpha \le 1$ and $q > |1-\alpha| > 0$. Now, we try to compare two upper bounds of the Tsallis relative entropy.
	
	Consider two bounds of the Tsallis relative entropy as 
	\begin{equation}
		B_{1}(\rho || \sigma) = D_{\alpha,z}^{(r)}(\rho || \sigma) \text{~and~} B_{2}(\rho || \sigma) =-\tr \left[\rho \ln_{1-\alpha}\left(\rho^{-q/2}\sigma^q\rho^{-q/2}\right)^{1/q}\right].
	\end{equation}
	Now, we have the following result for commuting positive semi-definite operators.
	\begin{theorem}
		Let $\rho$ be a density operator  and $\sigma$ be a positive semi-definite operator, such that, they commute, then $B_{2}(\rho || \sigma) \leq B_{1}(\rho || \sigma).$
	\end{theorem}
	
	\begin{proof}
		Since, $\rho$ and $\sigma$ are commutative, therefore $B_{1}(\rho || \sigma)$ and $B_{2}(\rho || \sigma)$ reduced to $B_{1}'(\rho || \sigma)$ and $B_{2}'(\rho || \sigma)$, where
		$$ B_{1}'(\rho || \sigma) = \frac{1}{\alpha-1}\ln_{r}\left[ \tr\left( \rho^{\alpha} \sigma^{1-\alpha} \right)\right] \text{~and~} B_{2}'(\rho || \sigma) = \frac{1}{\alpha-1}\left[ \tr\left( \rho^{\alpha} \sigma^{1-\alpha} \right) -1 \right].$$
		Now from \autoref{b}, $\alpha < 1$ and $r < 1$  we have
		\begin{equation}
			\begin{split}
				\ln_{r}\left[\tr\left( \rho^{\alpha}\sigma^{1-\alpha}\right)\right] &\leq \tr\left( \rho^{\alpha}\sigma^{1-\alpha}\right)-1,\\
				\text{or}~
				\frac{1}{\alpha-1}\ln_{r}\left[ \tr\left( \rho^{\alpha} \sigma^{1-\alpha} \right)\right] &\geq \frac{1}{\alpha-1}\left[ \tr\left( \rho^{\alpha} \sigma^{1-\alpha} \right) -1 \right].
			\end{split}	
		\end{equation}
		Therefore, $ B_{1}(\rho || \sigma) \geq B_{2}(\rho || \sigma) $.
	\end{proof}
	
	Whereas considering $\rho$ and $\sigma$ are non-commutative positive semi-definite operators, we observe that the $B_1(\rho || \sigma)$ is a tighter bound of the Tsallis relative entropy than $B_2(\rho || \sigma)$. We have the following numerical observations. Below, we compare $B_1(\rho || \sigma)$ and $B_2(\rho || \sigma)$ for $ 0 \leq \alpha < 1$ in different cases: 
	
	\begin{enumerate}[label=\textbf{Case \arabic*:}, leftmargin=*]
		\item 
		Let both $\tr(\rho)$ and $\tr(\sigma)$ belong to $(0, 1)$. Following equations (\ref{equality_of_Tsalish_bound}) and (\ref{bound definition}), we observe that the Tsallis relative entropy $D_{\alpha}(\rho || \sigma)$ may be negative. If we choose $\rho$ and $\sigma$ randomly such that, $\tr(\rho) \leq \tr\left( \rho^{\alpha}\sigma^{1-\alpha} \right) $
		 then $D_{\alpha}(\rho || \sigma)$ will be negative. 
		
		\item Let $\tr(\rho)$ in $(0,1)$ and $\tr(\sigma)=1$. In this case, we have similar observation as Case 1.
		
		\item Let $\tr(\rho)=1$ and $\tr(\sigma)$ in $(0,1)$. In this case, both $B_1(\rho || \sigma)$ and $B_2(\rho || \sigma)$ are the upper bounds of the Tsallis relative entropy, but the relation between $B_1(\rho || \sigma)$ and $B_2(\rho || \sigma)$ depends on the choice of $\rho$ and $\sigma$.

		For example, take
		$\rho=
		\begin{bmatrix}
			0.6 & 0.3 \\
			0.3 & 0.4
		\end{bmatrix}$
		and
		$\sigma=
		\begin{bmatrix}
			0.4 & 0.2 \\
			0.2 & 0.4
		\end{bmatrix}$. 
		For $\alpha=0.01$, $r=0.7$, $z=2$ and $q=1.5$ we have $B_1=0.207025$ and $B_2=0.203554$ that is $B_1(\rho || \sigma) > B_2(\rho || \sigma)$. Again for $\alpha=0.99$, $r=0.7$, $z=2$ and $q=1.5$ we have $B_1(\rho || \sigma) = 0.252861$ and $B_2(\rho || \sigma) = 0.264951$ that is $B_2(\rho || \sigma) > B_1(\rho || \sigma)$. We plot the values of $B_1(\rho || \sigma)$ and $B_2(\rho || \sigma)$ with respect to $\alpha \in (0, 1)$ considering $r= 0.7$, $z=2$, $q= 1.5$, in \autoref{fig:sub1}.
		\begin{figure}[h]
			\centering
			\begin{subfigure}{0.45\textwidth}
				\centering
				\includegraphics[width=\linewidth]{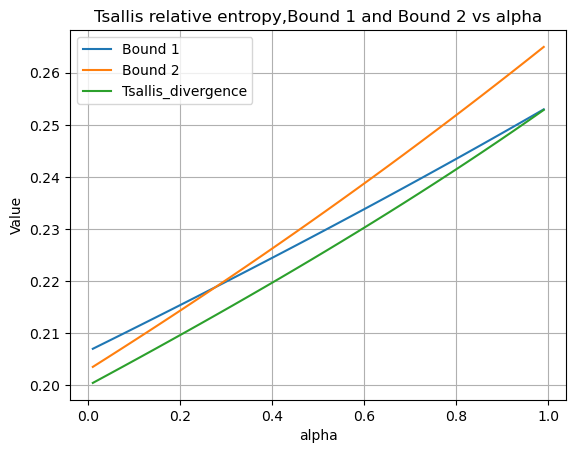}
				\caption{Sometime $B_1(\rho || \sigma)$ stays below $B_2(\rho || \sigma)$ and sometime $B_1(\rho || \sigma)$ stays above $B_2(\rho || \sigma)$.}
				\label{fig:sub1}
			\end{subfigure}
			\hfill
			\begin{subfigure}{0.45\textwidth}
				\centering
				\includegraphics[width=\linewidth]{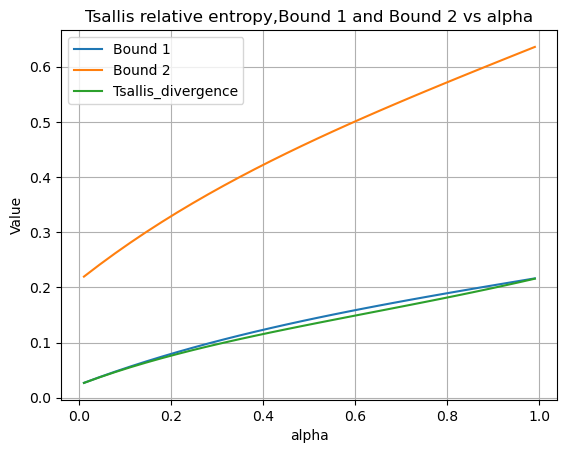}
				\caption{Here, $B_1(\rho || \sigma)$ provides a better bound to the Tsallis relative entropy than $B_2(\rho || \sigma)$.}
				\label{fig:sub2}
			\end{subfigure}
			\caption{For different choices of $\rho$ and $\sigma$ the bounds of Tsallis relative entropy behaves differently, when $\tr(\rho)=1$ and $0< \tr(\sigma) <1$.}
			\label{fig:main}
		\end{figure}

		Consider another examples. We consider $\rho$ and $\sigma$ as
		$\rho=
		\begin{bmatrix}
			0.2197 & 0.3956 \\
			0.3956 & 0.7803
		\end{bmatrix}$
		and
		$\sigma=
		\begin{bmatrix}
			0.1429 & 0.1677 \\
			0.1677 & 0.8334
		\end{bmatrix}$. Fixing $r= 0.7$, $z=2$, $q= 1.5$ we get $D_{\alpha}(\rho || \sigma)	\leq B_{1}(\rho || 	\sigma)\leq B_{2}(\rho || \sigma)$ when $\alpha \in (0, 1)$. This relation is plotted in \autoref{fig:sub2}.
		
		\item 
		Let $\tr(\rho)>1$ or $\tr(\sigma)>1$. In this case, we see that the Tsallis relative entropy $D_{\alpha}(\rho || \sigma)$ may be negative.
		For example take
		$\rho=
		\begin{bmatrix}
			0.8809 & -0.1665 \\
			-1.6665 & 0.90493
		\end{bmatrix}$
		and
		$\sigma=
		\begin{bmatrix}
			0.2857 & 0.3354 \\
			0.3354 & 1.666
		\end{bmatrix}$ with $\alpha=0.01$, we have $D_{\alpha}(\rho || \sigma)=-0.1582.$
				
		\item 
		Let $\tr(\rho)=1$ and $\tr(\sigma)=1$, that is both $\rho$ and $\sigma$ are density operators. Here, we numerically observe that for any two operators $\rho$ and $\sigma$, $B_1(\rho || \sigma)$ is a tighter bound of Tsallis relative entropy than $B_2(\rho || \sigma)$. For examples, consider two randomly chosen positive semi-definite operators
		$\rho=
		\begin{bmatrix}
			0.6 & 0.3 \\
			0.3 & 0.4
		\end{bmatrix}$
		and
		$\sigma=
		\begin{bmatrix}
			0.5 & 0.4 \\
			0.4 & 0.5
		\end{bmatrix}$.
		We plot $B_1(\rho || \sigma)$, $B_2(\rho || \sigma)$ and Tsallis relative entropy with respect to $\alpha \in (0, 1)$ for different values of $r \in \{0.5, 0.7, 0.8, 0.9\}$, $z=2$ and $q \in \{.991, 1\}$ in \autoref{tighter_bound_figure}.
		\begin{figure}
			\centering
			\includegraphics[scale = 0.5]{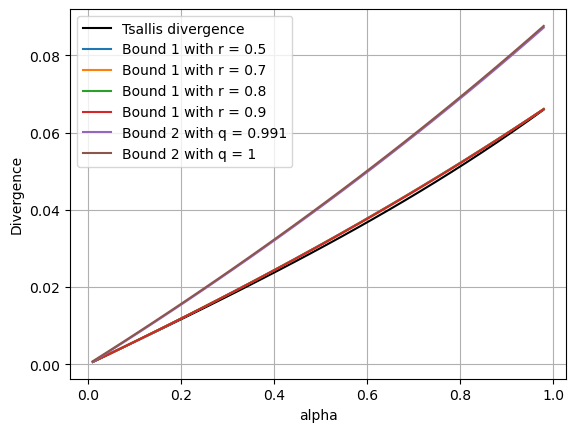}
			\caption{When $\rho$ and $\sigma$ are non-commutative density operators, $B_1(\rho||\sigma)$ acts as a bound of Tsallis relative entropy, which is tighter than $B_2(\rho||\sigma)$. } 
			\label{tighter_bound_figure}
		\end{figure}
	\end{enumerate}

	\section{Conclusion}
	The $\alpha$-$z$-R\'{e}nyi relative entropy was defined in \cite{audenaert2015alpha}. It utilizes the natural logarithm. As a drawback, the natural logarithm is not defined at $x=0$. In this work, we consider the $r$-logarithm in place of the natural logarithm in the expression of the $\alpha$-$z$-R\'{e}nyi relative entropy and formulate the $r$-deformed $\alpha$-$z$-R\'{e}nyi relative entropy $D_{\alpha,z}^{(r)}(\rho||\sigma)$. The $r$-logarithm is defined at $x=0$. In the limiting case, when $r \to 0$ we find the $\alpha$-$z$-R\'{e}nyi relative entropy from $D_{\alpha,z}^{(r)}(\rho||\sigma)$. We discuss different properties of $r$-logarithm, which are essential for our investigations. We also prove a number of characteristics of $D_{\alpha,z}^{(r)}(\rho||\sigma)$, which includes non-negativity, unitary invariance, pseudo-additivity, joint convexity, and the data-processing inequality. Interestingly, we observe that $D_{\alpha,z}^{(r)}(\rho||\sigma)$ is an upper bound of the Tsallis relative entropy. There are other known upper bounds of Tsallis relative entropy discussed in literature \cite{fujii2021matrix}. We justify that the new bound is tighter than the previous bound for the density operators. In future, the $\alpha$-$z$-R\'{e}nyi relative entropy may be used in different problems of classical and quantum information theory.
	

\end{document}